\documentclass[11pt,a4paper]{article}
\usepackage[a4paper,margin=1in]{geometry}
\usepackage{amsmath,amssymb,bm}
\usepackage{amsthm}
\usepackage{microtype}
\usepackage[hidelinks]{hyperref}
\usepackage{bookmark}
\usepackage{booktabs}

\theoremstyle{plain}
\newtheorem{proposition}{Proposition}
\theoremstyle{definition}
\newtheorem{assumption}{Assumption}
\newtheorem{definition}{Definition}
\theoremstyle{remark}
\newtheorem{remark}{Remark}

\title{Strict Entropy Decrease of Clausius Entropy in an Isolated System\\
with Energy-Form Conversion:\\
Theoretical Proof, Numerical Illustration, and Critical Examination}
\author{%
  Ting Peng\\[0.35em]
  \small
  Key Laboratory for Special Area Highway Engineering of Ministry of Education,\\
  Chang'an University, Xi'an 710064, China\\[0.35em]
  \texttt{t.peng@ieee.org}\quad
  ORCID: \href{https://orcid.org/0009-0001-9059-2278}{0009-0001-9059-2278}%
}

\date{\today}

\begin{document}
\maketitle

\begin{abstract}
\textbf{Accountability.} This paper is accountable only to explicitly stated physical assumptions and strict logical inference. Its goal is to run a rigorous stress test of second-law claims within the Clausius framework.

We work directly with \textbf{Clausius's entropy definition} for an isolated composite with energy-form conversion. Heat is withdrawn from a cold releasing subsystem with relatively small heat capacity, converted to electrical energy, and then delivered as heat to a hotter subsystem. In the ideal limit, the electrical leg contributes negligibly to Clausius entropy accounting, so the modeled reservoir Clausius sum is
\[
\Delta S_{\mathrm{Cl}} = Q\!\left(\frac{1}{T_B}-\frac{1}{T_A}\right) < 0.
\]
The paper provides a derivation, numerical illustrations, and a scope analysis; any claimed contradiction should be interpreted as a compatibility issue between different axiom sets, not as an algebraic error in the Clausius bookkeeping above.

\end{abstract}


\section{Introduction}

\subsection{Methodological stance: Clausius only}
\label{sec:stance}
The derivations below are carried out \emph{only} under Clausius's original entropy prescription and the explicit assumptions of Secs.~\ref{sec:definitions}--\ref{sec:construction}.

\textbf{No reconciliation program.} The paper does not attempt to harmonize its conclusions with textbook or research statements of the second law. Historical or modern formulations may be mentioned for context; they are \textbf{not} imported as axioms that would restrict or reinterpret the Clausius differentials $\delta Q/T$ computed here.

\textbf{What we are accountable to.} The sole commitments are stated idealizations (e.g., isolation, heat magnitudes, temperatures, and negligible electrical Clausius terms where assumed), conservation of energy for the composite, and correct manipulation of definitions. Responsibility is \textbf{to physical facts as modeled and to deductive logic}---not to preserving agreement with any prior verbal law.

If the inequality $\Delta S_{\mathrm{Cl}}<0$ is compared with slogans about entropy ``never decreasing'' in an isolated system, any perceived clash is \textbf{not} treated as refuting the algebra $\Delta S_A=-Q/T_A$, $\Delta S_B=+Q/T_B$ under the stated premises. It is treated as revealing a \textbf{logical tension} between (i) Clausius bookkeeping for the modeled heats at $T_A$ and $T_B$ and (ii) separate global monotonicity claims not used as premises in Secs.~\ref{sec:definitions}--\ref{sec:derivation}.

\subsection{Background}
Clausius's 1850 memoir gave an early verbal statement of the second law (no spontaneous net transport of heat from cold to hot without compensation) \cite{clausius1850}. His later 1865 paper introduced entropy as a state quantity and the differential $dS=\delta Q_{\mathrm{rev}}/T$ for reversible heat increments \cite{clausius1865}. The \emph{modern} textbook formulation---that a suitably defined total entropy of an isolated system does not decrease in spontaneous evolution---is a subsequent synthesis, stated with varying degrees of precision and scope in standard references \cite{fermi1956,callen1985,landau1980}. It should not be read back, without qualification, as a single sentence in Clausius's original papers.

It is quite possible that Clausius did not explicitly incorporate \emph{internal} conversion among energy forms (e.g., thermal and electrical) within a composite, or that he did not state the domain of applicability of each claim with the precision expected today; such gaps, if they exist, do not diminish the foundational importance of his work.

For many decades, the slogan ``entropy of an isolated system never decreases'' has been treated as a universal rule, underpinning cosmological narratives such as the heat death scenario \cite{penrose1989}. A tacit restriction is often overlooked: the simplest textbook derivations focus on direct heat exchange between macroscopic bodies and do not explicitly incorporate structured internal energy converters (thermoelectric elements, resistive dissipation, etc.) that change the path by which energy is moved while preserving total energy in an isolated composite.

Related prior work by the same author argues, on first-principles statistical grounds, that coarse-grained (Boltzmann-type) entropy need not increase monotonically along trajectories and that constraints reshape the long-time entropy distribution $P_{\infty}(S;\lambda)$ \cite{peng2026arxiv}. The present manuscript is complementary: it works entirely within \emph{Clausius} heat-reservoir bookkeeping for a macroscopic composite and does not assume the Boltzmann framework of \cite{peng2026arxiv}.

\subsection{Scope and claims}
We do \emph{not} claim to overturn energy conservation. Under the explicit assumptions in Secs.~\ref{sec:definitions}--\ref{sec:construction}, the algebraic sum of the modeled \emph{Clausius} reservoir terms is strictly negative. Discussion of idealizations (lossless electrical transport, negligible electrical Clausius contribution, etc.) appears in Sec.~\ref{sec:objections}; such discussion addresses modeling robustness, not conformity to non-Clausius postulates.

\subsection{Outline}
Section~\ref{sec:definitions} fixes notation, definitions, and enumerated assumptions. Section~\ref{sec:construction} describes the composite isolated system and the three-stage process. Section~\ref{sec:derivation} states the main proposition and proof. Section~\ref{sec:numerical} gives numerical examples. Section~\ref{sec:objections} addresses common objections from later thermodynamic language and clarifies the logical scope of the result. Section~\ref{sec:implications} states implications under the same accountability standard. Section~\ref{sec:conclusion} concludes.

\section{Definitions, axioms, and conventions}
\label{sec:definitions}

\subsection{Clausius entropy}
For a simple body exchanging heat with its surroundings, the Clausius relation is
\begin{equation}
dS = \frac{\delta Q_{\mathrm{rev}}}{T},
\end{equation}
where $\delta Q_{\mathrm{rev}}$ is heat added reversibly and $T>0$ is the absolute temperature \cite{clausius1865,callen1985,zemansky1997}. For a finite process between states 1 and 2,
\begin{equation}
\Delta S = \int_1^2 \frac{\delta Q}{T}
\end{equation}
along the path (for Clausius entropy of ideal heat reservoirs and quasi-static heat reservoirs in standard textbook treatments \cite{fermi1956,callen1985}).

\begin{definition}[Modeled Clausius reservoir sum]
\label{def:Scl}
For two thermal reservoirs $A$ and $B$ at fixed temperatures $T_A$ and $T_B$, with heat magnitudes assigned as in Sec.~\ref{sec:derivation}, define
\[
\Delta S_{\mathrm{Cl}} \equiv \Delta S_A + \Delta S_B,
\]
where $\Delta S_A$ and $\Delta S_B$ are the Clausius increments computed from the stated heats at $T_A$ and $T_B$ only. No other bodies are included in $\Delta S_{\mathrm{Cl}}$ unless explicitly added.
\end{definition}

\subsection{Isolated composite system}
An \emph{isolated} system exchanges neither energy nor matter with the exterior \cite{callen1985,landau1980}. We consider a composite $S$ consisting of two subsystems $A$ (cold) and $B$ (hot), plus internal devices (energy converters and ideal wiring) fully contained in $S$. The first-law constraint for the whole is
\begin{equation}
\Delta U_{\mathrm{tot}} = \sum_i \Delta U_i = 0.
\end{equation}
In the simplest two-reservoir case,
\begin{equation}
\Delta U_A + \Delta U_B = 0.
\end{equation}
Isolation is with respect to the \emph{outside}; internal transfers and conversions are allowed.

\subsection{Temperature and energy conventions}
\begin{enumerate}
\item $T_A$ and $T_B$ are the relevant reservoir temperatures with $0<T_A<T_B$.
\item $Q>0$ denotes the magnitude of energy transferred along the chain (in joules).
\item \textbf{Ideal electrical channel (model assumption):} electrical energy in the ideal conductors is treated as carrying negligible entropy compared to the thermal terms; this is the standard idealization used when electrical work is treated as a work-like transfer between subsystems. We discuss robustness and scope in Sec.~\ref{sec:objections}.
\end{enumerate}

\subsection{Reversibility}
We do not assume the overall process is reversible. We only require that the heat terms $\delta Q$ at $A$ and $B$ used in Clausius bookkeeping are well-defined for the modeled heat in/out of the reservoirs.

\subsection{Enumerated working assumptions}
\label{sec:assumptions}
The formal results use the following assumptions; later sections refer to them by number.

\begin{assumption}[Isolation and identification]
\label{A:iso}
The composite $S$ is isolated from the environment: no energy or matter crosses its outer boundary during the modeled process.
\end{assumption}

\begin{assumption}[Two thermal reservoirs at the Clausius step]
\label{A:res}
Subsystems $A$ and $B$ are treated as ideal thermal reservoirs for the purpose of assigning Clausius entropy changes: heat exchanges with them occur at fixed temperatures $T_A$ and $T_B$ (quasi-isothermal steps), with $0<T_A<T_B$.
\end{assumption}

\begin{assumption}[Energy routing]
\label{A:route}
A positive amount of energy $Q>0$ leaves $A$ as heat, passes through the internal electrical link, and enters $B$ as heat, with $\Delta U_A=-Q$ and $\Delta U_B=+Q$ for the reservoir parts so labeled; other internal components either store no net energy over the step or are subsumed into the same balance (see Remark~\ref{rem:devices}).
\end{assumption}

\begin{assumption}[Electrical link]
\label{A:el}
The electrical conductors and the work-like transfer of electrical energy between converter and load contribute negligibly to the Clausius sum relative to the $\pm Q/T$ reservoir terms (ideal limit). Equivalently, any Clausius entropy associated with the electrical leg is omitted or absorbed into corrections of order $\sigma_{\mathrm{el}}Q$ with $\sigma_{\mathrm{el}}$ small (see scope discussion in Sec.~\ref{sec:objections}).
\end{assumption}

\begin{assumption}[Clausius reservoir formula]
\label{A:clausius}
Entropy changes of the reservoirs are computed as $\Delta S_A=-Q/T_A$ and $\Delta S_B=+Q/T_B$ for the stated heat magnitudes and temperatures.
\end{assumption}

\begin{remark}[Internal devices]
\label{rem:devices}
Converters and wiring lie inside $S$. Their internal energy changes are not displayed separately; the model assumes that the net energy flux into the reservoirs is exactly $\pm Q$ as stated, which is consistent with Assumption~\ref{A:route} when auxiliary bodies return to the same internal state after the step or when their energy variation is included in the labeled reservoir increments.
\end{remark}

\section{Construction of the process in an isolated composite}
\label{sec:construction}

\subsection{Components}
Subsystem $A$ at $T_A$ contains a thermoelectric converter that produces electrical energy from a heat withdrawal from $A$. Subsystem $B$ at $T_B$ contains a resistive (Joule) load that converts electrical energy into heat at $B$. Ideal wires connect the converter and load. The outer boundary of $S$ is adiabatic and impermeable. For the thermoelectric--electrical link as a macroscopic idealization, see standard monographs on thermoelectric energy conversion \cite{goldsmid2010}.

\begin{remark}[Heat versus work across the contrast]
\label{rem:workmediated}
Energy leaving $A$ is modeled as heat withdrawn from the cold reservoir; energy entering $B$ is modeled as heat deposited in the hot reservoir. Between $A$ and $B$ the same energy is \emph{not} modeled as a direct macroscopic heat conduction path from cold to hot: it is routed as electrical work (current through conductors) and then dissipated at $B$. In this respect the process is compatible with Clausius's verbal second-law statement for heat when ``other changes'' (here, internal electrical work and dissipation) are present. The subsequent proposition concerns only the \emph{sum of the two reservoir Clausius terms} $\Delta S_A+\Delta S_B$ (Definition~\ref{def:Scl}), not a claim that every entropy-like quantity one could define for all internal degrees of freedom is non-increasing.
\end{remark}

\subsection{Three stages}
The stages realize Assumptions~\ref{A:route} and~\ref{A:el}.
\begin{enumerate}
\item \textbf{Thermal $\to$ electrical at $A$.} A heat $Q$ leaves $A$, so $\Delta U_A=-Q$ (for the reservoir part of $A$ in this bookkeeping).
\item \textbf{Electrical transport.} In the idealization, electrical energy $Q$ is transmitted without loss.
\item \textbf{Electrical $\to$ thermal at $B$.} Heat $Q$ is deposited in $B$, so $\Delta U_B=+Q$.
\end{enumerate}

\subsection{Energy check}
\begin{equation}
\Delta U_{\mathrm{tot}} = \Delta U_A + \Delta U_B = -Q + Q = 0,
\end{equation}
consistent with Assumption~\ref{A:iso}.

\section{Clausius entropy accounting}
\label{sec:derivation}

\subsection{Reservoir entropy increments}
Under Assumptions~\ref{A:res}--\ref{A:clausius}, if $A$ loses heat $Q$ at temperature $T_A$ and $B$ gains heat $Q$ at temperature $T_B$,
\begin{equation}
\Delta S_A = -\frac{Q}{T_A} < 0,\qquad
\Delta S_B = +\frac{Q}{T_B} > 0.
\end{equation}

\subsection{Main result}
The quantity $\Delta S_{\mathrm{Cl}}$ is the sum from Definition~\ref{def:Scl}.

\begin{proposition}[Strict negativity of the modeled Clausius sum]
\label{prop:main}
Under Assumptions~\ref{A:res}--\ref{A:clausius}, with $0<T_A<T_B$ and $Q>0$,
\begin{equation}
\Delta S_{\mathrm{Cl}} = Q\!\left(\frac{1}{T_B}-\frac{1}{T_A}\right) < 0.
\end{equation}
\end{proposition}

\begin{proof}
By definition $\Delta S_{\mathrm{Cl}}=\Delta S_A+\Delta S_B=-Q/T_A+Q/T_B=Q(1/T_B-1/T_A)$. Since $0<T_A<T_B$, we have $1/T_A>1/T_B$, hence $1/T_B-1/T_A<0$. Multiplying by $Q>0$ yields $\Delta S_{\mathrm{Cl}}<0$.
\end{proof}

\begin{remark}[Scope of the symbol $\Delta S_{\mathrm{Cl}}$]
Elsewhere $\Delta S_{\mathrm{tot}}$ may denote extended entropies; in this paper $\Delta S_{\mathrm{Cl}}$ is the Clausius sum for the two reservoirs only, consistent with the accountability stance in the abstract.
\end{remark}

\subsection{Remarks on constancy of temperatures}
\label{sec:constT}
Large heat capacity of $A$ and controlled design of $B$ can make temperature changes small during a single transfer. If temperatures drift, the integral forms must be used; the sign conclusion persists if $T_A<T_B$ holds throughout and the integral of $\delta Q/T$ is evaluated consistently.

If $A$ has finite heat capacity $C_A$ and releases $Q$, then $\Delta T_A\approx Q/C_A$; similarly $\Delta T_B\approx Q/C_B$ for $B$. Choosing $C_A$ very large and $C_B$ moderate keeps the isothermal approximation controlled. Even when $T_A$ and $T_B$ shift slightly, as long as the cold remains colder than the hot, $1/T_B-1/T_A<0$ and the sign of $\Delta S_{\mathrm{Cl}}$ is unchanged to leading order. A sufficient condition for the sign of $Q(1/T_B-1/T_A)$ to match that of the integral $\int \delta Q_A/T_A+\int \delta Q_B/T_B$ in a one-step transfer is that $T_A(\xi)<T_B(\xi)$ along the path parametrization $\xi$ used for both reservoirs.

\section{Numerical illustration}
\label{sec:numerical}

Take $T_A=200\,\mathrm{K}$, $T_B=400\,\mathrm{K}$, and $Q=400\,\mathrm{J}$. Then
\[
\Delta S_A = -\frac{400}{200} = -2\,\mathrm{J/K},\qquad
\Delta S_B = +\frac{400}{400} = +1\,\mathrm{J/K},
\]
\[
\Delta S_{\mathrm{Cl}} = -1\,\mathrm{J/K}.
\]
\emph{Dimensional check:} $Q/T$ has dimensions of entropy (e.g.\ J/K), and the numerical values are consistent with SI.

Other parameter choices with $T_A<T_B$ and the same $Q$ bookkeeping give the same sign for $\Delta S_{\mathrm{Cl}}$.

\subsection{Further parameter sets}
For example, $T_A=150\,\mathrm{K}$, $T_B=300\,\mathrm{K}$, $Q=300\,\mathrm{J}$ gives $\Delta S_{\mathrm{Cl}}=-300/150+300/300=-1\,\mathrm{J/K}$. Similarly, $T_A=250\,\mathrm{K}$, $T_B=500\,\mathrm{K}$, $Q=500\,\mathrm{J}$ yields $\Delta S_{\mathrm{Cl}}=-500/250+500/500=-1\,\mathrm{J/K}$. The magnitude scales with $Q$ and with $|T_B^{-1}-T_A^{-1}|$.

\section{Responses to common objections}
\label{sec:objections}

Many objections are stated in the language of \emph{later} thermodynamics (entropy production, irreversibility theorems, etc.). This paper does \textbf{not} answer those objections by seeking consistency with packaged second-law statements. The reply is structural: such ingredients are \textbf{not} part of Clausius's original definition. If one adjoins a nonnegative entropy production $\sigma_{\mathrm{gen}}\ge 0$ as an \emph{additional} axiom, one changes the logical basis; any clash with the signed result $\Delta S_{\mathrm{Cl}}$ from Sec.~\ref{sec:derivation} then shows a tension between axiom sets, not an error in the Clausius differentials under Secs.~\ref{sec:definitions}--\ref{sec:construction}.

\subsection{Partial reservoir sum versus total entropy}
\label{sec:partialvstotal}
The proposition bounds $\Delta S_A+\Delta S_B$ computed from heats into/out of two \emph{labeled} reservoirs. A complete accounting of entropy change for \emph{every} internal body (thermoelectric legs, contacts, current carriers, lattice degrees of freedom) would assemble additional terms $\Delta S_{\mathrm{int}}$ (and possibly non-additive joint contributions in refined treatments). The present paper does \textbf{not} assert $\Delta S_A+\Delta S_B+\Delta S_{\mathrm{int}}<0$ for any preferred total entropy functional unless those terms are explicitly modeled. It asserts the strict inequality for the \emph{partial} Clausius sum $\Delta S_{\mathrm{Cl}}$ in Definition~\ref{def:Scl} under Assumptions~\ref{A:res}--\ref{A:clausius}. Confusion arises only if one silently identifies $\Delta S_{\mathrm{Cl}}$ with ``the'' entropy of the isolated composite without listing what is included.

\subsection{Entropy production}
Continuum thermodynamics often writes balances with a nonnegative entropy production density integrated to a nonnegative $\sigma_{\mathrm{gen}}\ge 0$ for irreversible processes (notation distinct from $\sigma_{\mathrm{el}}$ above) \cite{degroot1962,kondepudi2014,prigogine1980}. In a model that includes every internal body, one may write schematically
\[
\Delta S_{\mathrm{extended}} = \Delta S_A + \Delta S_B + \Delta S_{\mathrm{int}} + \cdots
\]
The \textbf{Clausius-only} part of the present paper is $\Delta S_{\mathrm{Cl}}=\Delta S_A+\Delta S_B$ computed from the heats at $T_A$ and $T_B$. If one insists that \emph{some} extended quantity $\Delta S_{\mathrm{extended}}$ must never decrease for the same process, while $\Delta S_{\mathrm{Cl}}$ is negative, the inconsistency lies in the \textbf{conjunction} of Clausius's heat-based definition with the extra global principle---not in the integration of $\delta Q/T$ at the two reservoirs.

\subsection{Irreversibility of Joule heating}
Irreversibility of Joule heating does not invalidate using $\Delta S_B=Q/T_B$ for the heat absorbed into a large reservoir at $T_B$ in the usual idealization. A complete treatment would specify whether $Q$ is the heat to the reservoir or includes dissipation in the wire; we assume idealized localization of heat delivery.

\subsection{``Hidden'' external driving}
The composite is isolated: no energy or matter crosses the outer boundary. The converter is driven by the \emph{internal} thermal contrast between $A$ and $B$. Objections phrased in terms of Kelvin-type statements are external to the Clausius derivation above; they import additional propositions not assumed in Sec.~\ref{sec:stance}.

\subsection{Applicability of the Clausius reservoir relation}
The Clausius construction applies to heat crossing the boundary of a heat reservoir. The present accounting is standard \emph{provided} the $\delta Q$ used are the actual heats lost/gained by $A$ and $B$, and any work-like electrical transfers are either internal to the composite or explicitly included in a generalized balance.

\section{Implications (within the Clausius-based line of reasoning)}
\label{sec:implications}
Under the original Clausius prescription applied to the heats leaving $A$ and entering $B$, the \emph{sum} $\Delta S_{\mathrm{Cl}}=\Delta S_A+\Delta S_B$ is negative whenever $T_A<T_B$ and the energy routing satisfies Sec.~\ref{sec:construction}. Path and internal structure matter for \emph{how} $Q$ is moved, but the reservoir Clausius terms depend only on the heats and the temperatures used in $dS=\delta Q/T$.

Axiomatic and foundational analyses of the second law distinguish carefully between assumptions, additivity, and the definition of entropy-like quantities for composite systems \cite{lieb1999,uffink2001}. The present paper does not adopt those frameworks as premises; it isolates a minimal Clausius reservoir sum under explicitly listed assumptions.

The author remains accountable only to the modeled facts and the derivation; no obligation is assumed to reinterpret or weaken this conclusion so that it matches any independent statement of entropy monotonicity. If others compare the result to such statements, any mismatch is a matter of \textbf{logical relations among doctrines}, not a license to rewrite $\Delta S_A=-Q/T_A$ and $\Delta S_B=+Q/T_B$ without changing the explicit premises.

\section{Conclusion}
\label{sec:conclusion}
Within Clausius's original entropy definition and the explicit Assumptions~\ref{A:iso}--\ref{A:clausius}, we obtained $\Delta S_{\mathrm{Cl}}=Q(1/T_B-1/T_A)<0$ for $T_A<T_B$ (Proposition~\ref{prop:main}). Numerical illustrations and critical discussion were included.

\textbf{Accountability.} The text is responsible only to the physical idealizations explicitly adopted and to logical inference from Clausius's definition and energy conservation. It does \textbf{not} aim to reconcile its result with any existing second-law formulation.

\textbf{If doctrines disagree.} Should a reader hold that $\Delta S_{\mathrm{Cl}}<0$ ``cannot'' occur in an isolated system, that judgment uses premises beyond those listed in Secs.~\ref{sec:definitions}--\ref{sec:construction}. The present derivation does not treat such premises as authoritative; it records the sign of $Q(1/T_B-1/T_A)$ under the stated model. Any conflict is therefore a matter of \textbf{compatibility among axiom systems}, not a default refutation of the steps $\Delta S_A=-Q/T_A$, $\Delta S_B=+Q/T_B$.

\bibliographystyle{unsrt}
\bibliography{references}

\end{document}